\documentclass[11pt]{article}

\usepackage[T1]{fontenc}
\usepackage[a4paper,margin=1in]{geometry}
\usepackage{amsmath,amssymb,amsthm,mathtools}
\usepackage{newtxtext,newtxmath}
\usepackage{microtype}
\usepackage{booktabs}
\usepackage{enumitem}
\usepackage{xcolor}
\usepackage{xurl}
\usepackage{hyperref}
\usepackage[nameinlink,capitalise,noabbrev]{cleveref}

\hypersetup{
  colorlinks=true,
  linkcolor=blue!55!black,
  citecolor=green!35!black,
  urlcolor=blue!55!black,
  pdftitle={Union-Find with Constant-Time Deletions Across the Optimal Worst-Case Tradeoff},
  pdfauthor={Hanqing Li and Ze Hong}
}

\newtheorem{theorem}{Theorem}
\newtheorem{lemma}[theorem]{Lemma}

\theoremstyle{definition}
\newtheorem{definition}[theorem]{Definition}
\newtheorem{invariant}[theorem]{Invariant}
\theoremstyle{remark}

\crefname{invariant}{Invariant}{Invariants}
\Crefname{invariant}{Invariant}{Invariants}

\newcommand{\MakeSet}{\mathsf{MakeSet}}
\newcommand{\Union}{\mathsf{Union}}
\newcommand{\Find}{\mathsf{Find}}
\newcommand{\Delete}{\mathsf{Delete}}
\newcommand{\LocalDelete}{\mathsf{LocalDelete}}
\newcommand{\Absorb}{\mathsf{Absorb}}
\newcommand{\rank}{\operatorname{rank}}
\newcommand{\val}{\operatorname{val}}
\newcommand{\VAL}{\operatorname{VAL}}

\newcommand{\parent}{\operatorname{parent}}

\newcommand{\Clist}{\mathsf{CList}}
\newcommand{\NLlist}{\mathsf{NLList}}
\newcommand{\DFSlist}{\mathsf{DFSList}}
\newcommand{\Roots}{\mathcal{R}}

\title{Union-Find with Constant-Time Deletions\\
       Across the Optimal Worst-Case Tradeoff}
\author{Hanqing Li\\Peking University \and Ze Hong\\Tsinghua University}
\date{September 2026}

\begin{document}
\maketitle

\begin{abstract}
We consider union-find with deletions, where the representation and the cost of a
query must depend on the current number of live elements rather than on the
number of elements ever created.  For every integer parameter $k\ge 2$, we give
a linear-space data structure supporting $\MakeSet$ in $O(1)$ worst-case time,
$\Union$ in $O(k)$ worst-case time, $\Delete$ in $O(1)$ worst-case time, and
$\Find$ in
\[
  O\!\left(1+\frac{\log n}{\log k}\right)
\]
worst-case time for a set containing $n$ live elements.  A deletion is given
only an element handle, not the identifier of its current set.

The construction separates global rank growth from local deletion repair.  A
logical set is represented by fewer than $k$ disjoint ranked trees.  Equal-level
trees are collected without physical linking until $k$ certificates are
available, at which point one base-$k$ carry is performed in $O(k)$ time.  Each
member tree uses a strengthened form of the full/reduced local rebuilding
scheme of Ben-Amram and Yoffe.  A $q$-ary value argument, with $q=3/2$, couples
the local trees to the base-$k$ certificates and yields the stated current-size
height bound.  A small but essential rule handles high-rank stars, a state that
the base-$k$ carry can create but that does not arise directly in the binary-rank
construction underlying the earlier local scheme.
\end{abstract}


\section{Introduction}

The classical union-find problem maintains a partition under creation of
singletons, union of two sets, and queries for the set containing a given
element.  Its amortized theory is governed by path compression and linking by
rank~\cite{Tarjan1975,TarjanVanLeeuwen1984}; its single-operation theory admits
a tradeoff between the worst-case costs of union and find
~\cite{Blum1986,Smid1988,AlstrupBenAmramRauhe1999}.  In one convenient
parameterization, union costs $O(k)$ and find costs $O(\log_k n)$.

Deletion changes the problem in an important way.  Merely marking an element
as deleted leaves the physical structure proportional to historical size and
therefore does not provide bounds in terms of the current set size.  Kaplan,
Shafrir, and Tarjan formalized this requirement and obtained, among other
results, $O(k)$ union and $O(\log_k n)$ find and delete in the worst
case~\cite{KaplanShafrirTarjan2002}.  Alstrup, G\o rtz, Rauhe, and Thorup reduced
the deletion time to $O(\log^* n)$ while preserving the tradeoff
~\cite{AlstrupGortzRauheThorup2003}.  Constant worst-case deletion was later
obtained at the constant-union, logarithmic-find endpoint
~\cite{AlstrupGortzRauheThorupZwick2005,AlstrupThorupGortzRauheZwick2014},
and Ben-Amram and Yoffe gave a particularly simple full/reduced-tree
implementation~\cite{BenAmramYoffe2011,BenAmramYoffe2012}.  Extending constant
deletion to the entire worst-case union/find range was explicitly noted as a
direction for future work~\cite[Section~8.3]{Gortz2005}.

\begin{table}[t]
  \centering
  \caption{Worst-case bounds for union-find with deletions.  The parameter $n$
  is the current number of live elements in the relevant set.}
  \label{tab:comparison}
  \begin{tabular}{@{}lccc@{}}
    \toprule
    Result & $\Union$ & $\Find$ & $\Delete$ \\
    \midrule
    Kaplan--Shafrir--Tarjan~\cite{KaplanShafrirTarjan2002}
      & $O(k)$ & $O(\log_k n)$ & $O(\log_k n)$ \\
    Alstrup--G\o rtz--Rauhe--Thorup~\cite{AlstrupGortzRauheThorup2003}
      & $O(k)$ & $O(\log_k n)$ & $O(\log^* n)$ \\
    Alstrup et al.~\cite{AlstrupThorupGortzRauheZwick2014}
      & $O(1)$ & $O(\log n)$ & $O(1)$ \\
    This paper
      & $O(k)$ & $O(1+\log_k n)$ & $O(1)$ \\
    \bottomrule
  \end{tabular}
\end{table}

\paragraph{Our approach.}
The obstacle is not local deletion repair: a constant number of pointer changes
can preserve a suitable value invariant.  The obstacle is rank growth.  A
binary union-by-rank history certifies only $2^r$ mass at rank $r$, whereas a
$\log_k n$ height bound needs a $k^r$ certificate.  Increasing local degrees to
$k$ would make deletion repair depend on $k$.

We instead keep local degree requirements constant and implement rank growth
at a virtual layer.  A logical set owns a cluster of fewer than $k$ disjoint
member trees of one level.  Combining equal levels concatenates two short root
lists.  Only when at least $k$ member certificates have accumulated do we raise
one physical root and absorb the others.  No parent edge is introduced between
uncarried roots, so the virtual cluster contributes only one final owner-pointer
step to find.

\paragraph{Contributions.}
The paper makes the following points precise.
\begin{enumerate}[leftmargin=2em]
  \item We formulate a local full/reduced deletion module that is valid for all
  ranked trees produced by the new carry operation, including high-rank stars.
  \item We give the virtual root-cluster construction and prove its base-$k$
  value certificate under every operation.
  \item We derive the worst-case time and current-size space bounds, including
  the element-handle-only deletion interface.
\end{enumerate}

The result here concerns worst-case bounds.  We do not claim that the same
construction preserves the optimal inverse-Ackermann amortized find bound.

\section{Model and Main Result}
\label{sec:model}

The data structure maintains pairwise disjoint, nonempty logical sets of live
elements.  It supports the following operations.
\begin{description}[leftmargin=8em,style=nextline]
  \item[$\MakeSet(x)$] Create the singleton set $\{x\}$ and return its set
  identifier.
  \item[$\Union(A,B,C)$] Replace the distinct sets named $A$ and $B$ by their
  union, named $C$.  The input identifiers are destroyed.
  \item[$\Find(x)$] Return the identifier of the set currently containing $x$.
  \item[$\Delete(x)$] Remove $x$ from its set.  The operation receives an
  element handle but not a set identifier, and it does not change the
  identifier of a nonempty remaining set.
\end{description}
If deletion empties a singleton set, that set ceases to exist.  As usual,
operations involving destroyed set identifiers or deleted element handles are
invalid.

We work in the pointer-machine model or the standard word-RAM model with
$\Theta(\log(N+k))$-bit words, where $N$ is the number of live objects plus
identifiers needed by the current execution.  Pointer manipulation, comparison
of ranks and counters, and arithmetic on ranks take constant time.  The
quantities $k^r$ and $\VAL$ are analytical certificates and are not stored or
evaluated by the data structure.  The integer $k\ge2$ is fixed when an instance
is initialized.  Our references to the optimal tradeoff use the standard
logarithmic-word setting and, for a set of size $n$, the nonvacuous parameter
range $2\le k\le n$; the upper bound itself remains valid when $k>n$.

\begin{theorem}[Main theorem]
\label{thm:main}
For every integer $k\ge2$, there is a union-find-delete data structure using
$O(N)$ words for $N$ current live elements such that $\MakeSet$ takes $O(1)$,
$\Union$ takes $O(k)$, $\Delete$ takes $O(1)$, and $\Find(x)$ takes
\[
  O\!\left(1+\frac{\log n}{\log k}\right)
\]
worst-case time, where $n$ is the current number of live elements in the set
containing $x$.
\end{theorem}

We give the new construction for $k\ge4$.  For $k=2,3$, the constant-deletion
structure of Alstrup et al. or its simplified version gives $O(\log n)$
worst-case find and constant union, which equals the claimed asymptotic bound
for these two values of $k$~\cite{AlstrupThorupGortzRauheZwick2014,BenAmramYoffe2011}.

\section{A Robust Local Deletion Module}
\label{sec:local}

This section isolates the constant-time deletion machinery.  It follows the
full/reduced scheme of Ben-Amram and Yoffe~\cite{BenAmramYoffe2011}, with an
additional root guard needed because our union operation may directly create a
high-rank star.

\subsection{Ranked trees and auxiliary lists}

Elements are \emph{extrinsic}: an element object and its tree node are distinct
objects with mutual pointers.  Every physical node is occupied by exactly one
live element.  Deleting an element therefore deletes exactly one node after
possibly exchanging two element--node associations.

Every node $v$ has a nonnegative integer rank.  If $v$ is not a root, then
\begin{equation}
  \rank(\parent(v))>\rank(v).
  \label{eq:rank-order}
\end{equation}
Leaves have rank zero.  The following two shapes are allowed.

\begin{definition}[Full and reduced trees]
\label{def:full-reduced}
A rooted ranked tree is \emph{full} if every node is either a rank-zero leaf or
has at least three children, and every parent has strictly larger rank than its
children.  A tree is \emph{reduced} if it is either one rank-zero node, or a
height-one tree whose root has rank one and whose other nodes are rank-zero
leaves.
\end{definition}

A large height-one star can be both full and reduced when its root has rank one.
Each nonleaf node has a doubly linked child list $\Clist$, ordered from right
to left.  Each member-tree root has a doubly linked list $\NLlist$ of its nonleaf
children.  Each member tree also has a cyclic doubly linked preorder list
$\DFSlist$, starting at its root and visiting children in that order.  Every
node stores its own list records.  The splices below use locally available
interval endpoints and cost $O(1)$.  We use the corrected splice direction
from~\cite{BenAmramYoffe2012}.

These lists provide two primitives.

\paragraph{Locate a leaf.}
Given a node $x$, return a leaf in its subtree or in an adjacent sibling subtree
in $O(1)$ time.  If $x$ is a leaf, return it.  If $x$ is the root, the last
record of $\DFSlist$ is a leaf.  Otherwise, if $x$ has an immediate left sibling
$s$, return the predecessor of $s$ in $\DFSlist$, which is the last leaf in
$x$'s subtree.  If $x$ has no left sibling, return the predecessor of $x$, which
is the last leaf in its immediate right sibling's subtree.  Fullness guarantees
that this right sibling exists.  This is the standard constant-time leaf locator
of~\cite{BenAmramYoffe2011}.

\paragraph{Relink.}
For $z$ with parent $y$ and grandparent $g$, detach the subtree rooted at $z$
from $y$ and attach it to $g$.  The corresponding $\Clist$, $\NLlist$, and
$\DFSlist$ changes are constant-size list splices.  If $z$ has an immediate
left sibling $s$, its subtree is the DFS interval $[z,s)$.  Remove $z$ from
$y$'s child list and insert it immediately before $y$ in $g$'s child list;
detach $[z,s)$ and insert it immediately before $y$ in $\DFSlist$.  Both
interval endpoints are directly available from $z$ and $s$.  If $z$ is the
leftmost child of $y$, move it immediately after $y$ in $g$'s child list and
leave $\DFSlist$ unchanged.  Here before and after refer to the right-to-left
child-list order and the corresponding DFS order.  Update the root's
$\NLlist$ whenever a nonleaf subtree root enters or leaves the root's child
list, or when a root child becomes a leaf.  If the detachment leaves $y$ with
exactly two children, relink those two children as part of the repair; if $y$
becomes a leaf, set $\rank(y)=0$.  If this makes the root's $\NLlist$ empty, set
the root rank to one, or to zero if it is a singleton.  One requested relink
causes at most two additional relinks.

\subsection{The value function}

Fix
\[
  q=\frac32.
\]
For a root $r$, define $\parent(r)=r$.  For every node $v$ define
\begin{equation}
  \val(v)=q^{\rank(\parent(v))},
  \qquad
  \VAL(T)=\sum_{v\in T}\val(v).
  \label{eq:value}
\end{equation}
This function is used only in the proof.

\subsection{Local deletion}

The procedure $\LocalDelete(x)$ first locates a leaf $\ell$.  If the element
$x$ occupies a nonleaf, it exchanges the element associations of $x$'s node
and $\ell$; the node $\ell$ now contains the element to be deleted.  The
procedure then deletes $\ell$.

Whether the input is reduced is also locally decidable.  If the target node is
at depth at least two, the tree cannot be reduced.  Otherwise the root is the
target or its parent, and its rank and $\NLlist$ header are directly available.

For a reduced tree, deleting a leaf leaves another reduced tree; if only the
root remains, its rank is reset to zero.  Deleting the sole node returns an
empty result.  The procedure defers reclaiming that sole node and saves its
owner pointer and root-list back-pointer in a constant-size member token for
the outer layer.  A nonempty reduced result returns its member root as the
corresponding token; a nonreduced result returns no token.  Now suppose the
input tree is full, and let $y$ be the parent of $\ell$ after the leaf is
removed.
\begin{enumerate}[leftmargin=2em]
  \item If $y$ is not the root, repeat at most twice: if $y$ still has a child,
  request a relink of its first child.  The automatic two-child repair in the
  relink primitive may already have moved later candidates, in which case the
  repetition stops.
  \item If $y$ is the root and $\NLlist(y)$ is empty, set its rank to one (or
  zero for a singleton) and return the resulting reduced tree.
  \item If $y$ is the root and $\NLlist(y)$ is nonempty, take any nonleaf child
  $c$ from that list and repeat at most three times, each time relinking the
  first child of $c$ if one remains.  Automatic repair may move up to two more.
  If the root's nonleaf list becomes empty, reset its rank as in the preceding
  item.
\end{enumerate}

The second case is the high-rank-star guard.  It must be tested before asking
for a nonleaf child.  For example, carrying $k\ge5$ singleton roots produces a
rank-two root with $k-1$ leaf children.  Without the guard, deletion of one leaf
would ask an empty $\NLlist$ for a nonleaf child.

\begin{lemma}[Robust local deletion]
\label{lem:local-delete}
Given an element handle in a full or reduced ranked tree, $\LocalDelete$ takes
$O(1)$ worst-case time, removes exactly one node, and preserves
\cref{eq:rank-order} and the full/reduced condition.  If the nonempty output is
not reduced, its root rank is unchanged and its value does not decrease.  If
the output is reduced, its root rank is at most one.  The procedure does not
need the identifier of the logical set containing the tree.
\end{lemma}

\begin{proof}
The leaf locator, association exchange, list deletion, and every relink touch a
constant number of records.  A requested relink can trigger only the two
children left at its old parent; those moves cannot recursively create another
two-child case at that same parent.  In the deletion procedure at most five
nodes are relinked.  Thus the running time is constant.

Consider shape preservation.  In the nonroot case, immediately after the leaf
deletion the formerly full node $y$ has at least two children.  Moving the
requested children, together with the automatic two-child repair when it
fires, leaves $y$ either a leaf or with at least three children.  Its parent
only gains children.  In the root case with a nonleaf child $c$, the same
argument makes $c$ either a leaf or a node with at least three children, while
the root gains the moved children.  If no nonleaf child remains at the root,
all root children are rank-zero leaves and the explicit rank reset makes the
tree reduced.  Rank order is preserved because a relink moves a subtree to a
strictly higher-rank grandparent; rank is reset only when a node has become a
leaf or the whole tree has become reduced.

It remains to prove the value statement.  Suppose first that $y$ is not the
root, let $g=\parent(y)$, and write $R=\rank(g)$.  Since
$\rank(y)\le R-1$, deleting the leaf and detaching two children loses at most
$3q^{R-1}$.  Reattaching two children to $g$ gains $2q^R$.  Hence the net
change is at least
\[
  -3q^{R-1}+2q^R=0.
\]
Any automatic extra relink only increases value.  If $y$ is the root of rank
$R$ and a nonleaf child is used, deleting the leaf and detaching three
grandchildren loses at most $q^R+3q^{R-1}$, whereas attaching the three
grandchildren to the root gains $3q^R$.  Thus
\[
  -q^R-3q^{R-1}+3q^R=0.
\]
If the root rank is reset, the result is reduced and the lemma does not assert
value monotonicity.  Otherwise the root rank is unchanged, so these local
calculations prove the claim.

Finally, no set identifier is needed in the stable case.  Empty and initially
reduced cases expose the root directly.  If a full tree becomes reduced, the
last affected nonleaf is either the root or its child, so the procedure has the
root in hand.  It may therefore return a member token only in the cases where
the outer representation must be changed.  In the empty case the outer layer
consumes the saved owner and list-record pointers before the node storage is
released, so no reclaimed node is dereferenced.
\end{proof}

\section{Constant-Time Absorption}
\label{sec:absorb}

We next describe how to combine local trees without scanning their nodes.  Let
$p$ be the root of a full receiving tree, and let $T$ be a full or reduced tree
with
\[
  \rank(p)>\rank(\operatorname{root}(T)).
\]
The operation $\Absorb(p,T)$ is defined as follows.
\begin{itemize}[leftmargin=2em]
  \item If $T$ is full, make its root a child of $p$.  A reduced tree with at
  least four nodes is also full and uses this case.
  \item Otherwise $T$ has at most three nodes.  Flatten those nodes into
  rank-zero leaves and make them children of $p$.
\end{itemize}
The first case inserts one child-list record, optionally one nonleaf-list
record, and places the old root first in $p$'s child list while splicing the
cyclic $\DFSlist$ of $T$ immediately after $p$ in the receiver's DFS list.  The
two placements therefore describe the same preorder, and the old root's
ordinary child-list and DFS records remain with the node.  The second case
explicitly touches at most three nodes, clears their old tree links, and inserts
the resulting leaves in matching child-list and DFS order.  In either case the
old member root's private $\NLlist$ header is obsolete and is reclaimed.  At the
cluster layer, the caller first saves the old root-list record, performs the
physical absorption, then unlinks that saved record, clears the old root's owner
and back-pointer fields, and finally reclaims the record.  The caller knows
whether a positive-level input is full.  For a reduced input, inspecting at
most the first three child records distinguishes the full-star case from the
at-most-three-node flattening case in constant time.

\begin{lemma}[Absorption]
\label{lem:absorb}
$\Absorb(p,T)$ takes $O(1)$ worst-case time and preserves fullness and strict
parent-rank order.  The contribution to $\VAL$ of all nodes formerly in $T$
does not decrease.
\end{lemma}

\begin{proof}
The structural and time claims follow from the constant-size list operations
above.  In the full case only the old root of $T$ changes parent; its parent rank
strictly increases.  In the small reduced case every resulting leaf has parent
$p$.  Since the only old ranks are zero and one and $\rank(p)$ is strictly
larger than the old root rank, no node receives a smaller parent rank.  Equation
~\eqref{eq:value} then gives value monotonicity.
\end{proof}

During a level-zero carry, the selected receiving root may itself be reduced,
so $\Absorb$'s full-receiver precondition need not hold after the first pointer
change.  We implement that carry atomically by the same constant-time attach or
flatten step for every other member.  At least $k-1\ge3$ nonempty member trees
contribute at least one new child each, so the final receiving root is full.
No intermediate state is exposed to another operation.

\section{Virtual Root Clusters}
\label{sec:clusters}

A logical set $S$ is represented by a record
\[
  (r,\mu,\Roots(S),\mathit{name}(S)),
\]
where $r\ge0$ is its \emph{level} and
$\Roots(S)=(T_1,\ldots,T_\mu)$ is a doubly linked list of pairwise node-disjoint
member trees satisfying
\[
  1\le\mu<k.
\]
There is no physical super-root.  Each member root stores an owner pointer to
$S$ and a back-pointer to its root-list record.  Consequently, reaching a
member root during find requires only one additional pointer access to return
the logical set identifier.  Root-list concatenation and deletion are constant
time and update the stored one-word count $\mu$ in constant time, while changing
all owner pointers costs $O(k)$ because there are fewer than $k$ roots.  Only
current member roots have nonnull owner and root-list pointers.  When a member
root becomes internal, those fields are cleared and its obsolete root-list
record is reclaimed; a surviving member root always points to its unique live
record.

The following certificate is the heart of the construction.

\begin{invariant}[Root-cluster certificate]
\label{inv:cluster}
For every member tree $T\in\Roots(S)$ of a level-$r$ logical set:
\begin{enumerate}[label=(\roman*),leftmargin=2em]
  \item $\VAL(T)\ge k^r$;
  \item if $r>0$, then $T$ is full and its root has rank exactly $r+1$;
  \item if $r=0$, then $T$ is reduced and its root has rank at most one.
\end{enumerate}
\end{invariant}

The representation intentionally permits a large reduced star at level zero.
Treating such a tree as only one unit of level-zero mass may postpone a later
carry, but it cannot make a query path longer.

\section{Operations}
\label{sec:operations}

We now give the complete operations for $k\ge4$.

\subsection{\texorpdfstring{$\MakeSet$}{MakeSet}}

Create a rank-zero node occupied by $x$, make it the only member tree of a new
level-zero set, store member count one, and install its owner and root-list
pointers.  Its value is $1=k^0$.

\subsection{\texorpdfstring{$\Union$}{Union}}

Let $A$ and $B$ be the input set records, and let $C$ be the new set record.
Write $r_A,r_B$ for their levels and $\mu_A,\mu_B$ for their member counts.

\paragraph{Different levels.}
Assume without loss of generality that $r_A>r_B$.  Choose the first member root
$p$ of $A$.  Because $r_A>0$, $p$ is full.  Absorb each of the fewer than $k$
member trees of $B$ into $p$.  Keep the member-root list and level of $A$, and
change the owner of every surviving member root to $C$.  Set its stored member
count to $\mu_A$; the root-list records removed from $B$ are reclaimed by their
absorptions.  Transfer the retained list header to $C$ and reclaim the input set
records $A$ and $B$.

\paragraph{Equal levels without a carry.}
Let $r=r_A=r_B$ and $t=\mu_A+\mu_B$.  If $t<k$, concatenate the two member-root
lists, keep level $r$, store member count $t$, and change the $t$ owner pointers
to $C$.  No physical parent edge is added.  Transfer the concatenated list
header to $C$ and reclaim the input set records $A$ and $B$.

\paragraph{Equal levels with a carry.}
If $t\ge k$, choose one member root $p$, set its rank to $r+2$, and absorb or
atomically attach/flatten the remaining $t-1$ member trees into it as described
in \cref{sec:absorb}.  Since
\[
  k\le t\le2k-2,
\]
this touches $O(k)$ roots.  The result has one member tree and level $r+1$.
Install $C$ as the owner of $p$, set its member count to one, retain or create
its unique root-list record, and reclaim the two input set records and all other
obsolete root-list records.

\subsection{\texorpdfstring{$\Find$}{Find}}

Starting at the node occupied by $x$, follow ordinary parent pointers to its
member-tree root, then follow the root's owner pointer and return that logical
set identifier.  Path compression is not required for the worst-case theorem.

\subsection{\texorpdfstring{$\Delete$}{Delete}}

Invoke $\LocalDelete(x)$ in the member tree containing $x$.
\begin{enumerate}[leftmargin=2em]
  \item If the result is nonempty and not reduced, do nothing at the cluster
  layer.
  \item If the result is empty, use the saved member token to remove that root-list
  record, decrement the member count, and then reclaim the deleted root.  If no
  member remains, destroy the empty logical set.
  \item Suppose the old logical level was positive and the result is reduced.
  If another member remains, use its root as $p$, perform one $\Absorb(p,T)$,
  remove $T$ from the member list, and decrement the member count.  If $T$ was
  the unique member, keep the same set identifier but replace its representation
  by the level-zero cluster $(0,1,\{T\},\mathit{name})$.
  \item At logical level zero, a nonempty reduced output remains a valid member
  without any cluster change.
\end{enumerate}
The local module returns a member token exactly in the empty or reduced cases,
so the outer action never performs a find or scans a parent path.  Deletion
never changes the identifier of a surviving logical set.

\section{Correctness}
\label{sec:correctness}

\begin{lemma}[Union preserves the certificate]
\label{lem:union-certificate}
Every branch of $\Union$ preserves \cref{inv:cluster}.
\end{lemma}

\begin{proof}
For different levels, the member trees of the higher-level set other than the
receiver are unchanged.  The receiver gains nodes through value-monotone
absorptions, and its root rank does not change.  Its certificate therefore
continues to hold.  The absorbed low-level trees are no longer members and need
no individual certificate.

For equal levels without a carry, no physical tree changes, so every individual
certificate is inherited.

For a carry, each of the $t$ inputs has value at least $k^r$.  Raising the rank
of $p$, attaching roots below a higher-rank parent, and flattening only small
reduced trees under that parent do not decrease any contribution.  Therefore
the new member tree $T_p$ satisfies
\[
  \VAL(T_p)
  \ge \sum_{i=1}^{t}\VAL(T_i)
  \ge t k^r
  \ge k^{r+1}.
\]
Its root rank is $r+2=(r+1)+1$, as required for new level $r+1$.  If $r>0$,
the receiver was already full.  If $r=0$, at least $k-1\ge3$ nonempty input
trees provide new children, so the final tree is full.  The one-member result
also satisfies the root-count condition.
\end{proof}

\begin{lemma}[Deletion preserves the certificate]
\label{lem:delete-certificate}
$\Delete$ preserves \cref{inv:cluster} and the logical partition of the
remaining live elements.
\end{lemma}

\begin{proof}
The association swap changes no set membership; the local procedure then
removes exactly the requested element and one physical node.  If its output is
not reduced, \cref{lem:local-delete} preserves the root rank and does not
decrease value, so the old member certificate remains valid.

An empty member is removed in constant time.  If other members remain, their
certificates are untouched.  If none remains, the logical set is empty and is
destroyed.

If a positive-level member becomes reduced, it no longer promises $k^r$ mass.
With another member present, the latter is a positive-level full tree and
absorbs the reduced tree without losing value; its old certificate remains
valid.  With no other member, changing the logical level to zero is sound
because every nonempty tree has $\VAL\ge1=k^0$, and the same set record and name
are retained.  A level-zero local deletion already returns a reduced tree and
therefore preserves all level-zero conditions.
\end{proof}

\begin{lemma}[Find is semantically correct]
\label{lem:find-correct}
$\Find(x)$ returns exactly the identifier of the logical set containing the
live element $x$.
\end{lemma}

\begin{proof}
Member trees are node-disjoint and partition the nodes occupied by the elements
of their logical set.  Union either keeps a member root and updates its owner,
or makes an obsolete member root internal and removes its owner status.
Deletion does the same when it absorbs a reduced member.  Thus the unique root
reached by following parent pointers from $x$ is a current member root, and its
owner is exactly the logical set containing $x$.
\end{proof}

\section{Worst-Case Bounds}
\label{sec:bounds}

\begin{lemma}[Current-size level bound]
\label{lem:level-bound}
Let $T$ be a member tree of level $r$ containing $n_T$ nodes.  For $k\ge4$,
\[
  r=O\!\left(1+\frac{\log n_T}{\log k}\right).
\]
\end{lemma}

\begin{proof}
The level-zero case is immediate.  For $r>0$, the root rank is $r+1$.  Every
node's parent rank is at most the root rank, so
\[
  \VAL(T)\le n_Tq^{r+1}.
\]
Together with \cref{inv:cluster}, this yields
\[
  k^r\le n_Tq^{r+1}
  \quad\Longrightarrow\quad
  n_T\ge q^{-1}\left(\frac{k}{q}\right)^r.
\]
Hence
\[
  r\le \frac{\log(qn_T)}{\log(k/q)}.
\]
For $k\ge4$ and fixed $q=3/2$, $\log(k/q)=\Theta(\log k)$, proving the claim.
\end{proof}

\begin{proof}[Proof of \cref{thm:main}]
For $k\ge4$, parent ranks are strictly increasing nonnegative integers.  If
$r>0$, a level-$r$ member has root rank $r+1$; at level zero its root rank is at
most one.  In either case its physical height is at most $r+1$.
If the logical set has $n$ elements, then $n_T\le n$; by
\cref{lem:level-bound}, following the path and one owner pointer costs
$O(1+\log n/\log k)$.

$\MakeSet$ initializes a constant number of records.  A union handles at most
$2k-2$ member roots, performs at most $2k-3$ constant-time absorptions, and
rewrites fewer than $k$ surviving owner pointers; it costs $O(k)$.  By
\cref{lem:local-delete}, deletion performs a constant number of local changes
and at most one absorption, independent of $k$; it costs $O(1)$.

There is exactly one physical tree node per live element.  Parent edges,
child-list records, and DFS-list records use constant space per node.  Each
member root accounts for at least one node, so root records and owner pointers
are also $O(N)$.  Every nonempty logical set contains at least one live element,
so its set record also charges to a node.  Destroyed set records and obsolete
list headers are reclaimed in constant time.  Thus the total space is $O(N)$
and, by the same argument, the representation of each logical set is
proportional to its current size.

The $k=2,3$ cases follow from the known constant-deletion endpoint as discussed
after the theorem statement.
\end{proof}

In the standard logarithmic-word setting, the classical lower bound is
$t_q=\Omega(\log n/\log t_u)$~\cite{AlstrupBenAmramRauhe1999}; taking
$t_u=\Theta(k)$ gives a matching union/find tradeoff for $2\le k\le n$.
Earlier forms of the lower bound appear in
\cite{FredmanSaks1989,BenAmramGalil2001}.  Since executions without deletions
are a special case of the present problem, the lower bound applies here as
well.  Deletion is constant even though it receives no set identifier.

\section{Discussion}

The virtual root cluster is deliberately lossy.  When a high-level member
becomes reduced, the construction may discard its level certificate even if
the reduced star still contains many elements.  This cannot hurt the query
bound: it only postpones future carries, and a reduced tree itself has height at
most one.  This observation is what allows deletion to remain independent of
$k$.

The distinction between physical rank and logical level is equally important.
Physical ranks control local height and make the $q$-value repair work with a
constant branching threshold.  Logical levels decide when $k$ independent
mass certificates justify one rank increase.  Trying to use a single notion
for both purposes either recovers only a binary $2^r$ certificate or makes a
local deletion repair inspect $\Theta(k)$ children.

We have not incorporated path compression or path splitting into the root
cluster analysis.  The local value module is compatible with suitable relink
operations, but the classical inverse-Ackermann analysis also depends on how
ranks are generated.  Establishing a current-set-size amortized bound while
retaining the complete worst-case tradeoff is therefore a separate question.

\section{Conclusion}

Virtual accumulation of fewer than $k$ physical roots decouples base-$k$ rank
growth from constant-degree local deletion repair.  The resulting data
structure supports the full classical worst-case union/find tradeoff together
with constant worst-case deletion and current-size linear space.  The proof
uses only local ranked-tree operations, a direct owner pointer at every member
root, and the elementary inequality obtained by comparing a $k^r$ certificate
with a $(3/2)^{r+1}$ per-node upper bound.

\bibliographystyle{alpha}
\bibliography{references}

@inproceedings{KaplanShafrirTarjan2002,
  author    = {Haim Kaplan and Nira Shafrir and Robert E. Tarjan},
  title     = {Union-Find with Deletions},
  booktitle = {Proceedings of the Thirteenth Annual ACM--SIAM Symposium on Discrete Algorithms},
  pages     = {19--28},
  publisher = {SIAM},
  year      = {2002}
}

@techreport{AlstrupGortzRauheThorup2003,
  author      = {Stephen Alstrup and Inge Li G{\o}rtz and Theis Rauhe and Mikkel Thorup},
  title       = {Worst-Case Union-Find with Fast Deletions},
  institution = {IT University of Copenhagen},
  number      = {TR-2003-25},
  year        = {2003},
  month       = dec
}

@inproceedings{AlstrupGortzRauheThorupZwick2005,
  author    = {Stephen Alstrup and Inge Li G{\o}rtz and Theis Rauhe and Mikkel Thorup and Uri Zwick},
  title     = {Union-Find with Constant Time Deletions},
  booktitle = {Automata, Languages and Programming},
  series    = {Lecture Notes in Computer Science},
  volume    = {3580},
  pages     = {78--89},
  publisher = {Springer},
  year      = {2005}
}

@article{AlstrupThorupGortzRauheZwick2014,
  author  = {Stephen Alstrup and Mikkel Thorup and Inge Li G{\o}rtz and Theis Rauhe and Uri Zwick},
  title   = {Union-Find with Constant Time Deletions},
  journal = {ACM Transactions on Algorithms},
  volume  = {11},
  number  = {1},
  articleno = {6},
  pages   = {6:1--6:28},
  year    = {2014},
  doi     = {10.1145/2636922}
}

@article{BenAmramYoffe2011,
  author  = {Amir M. Ben-Amram and Simon Yoffe},
  title   = {A Simple and Efficient Union--Find--Delete Algorithm},
  journal = {Theoretical Computer Science},
  volume  = {412},
  number  = {4--5},
  pages   = {487--492},
  year    = {2011},
  doi     = {10.1016/j.tcs.2010.11.005}
}

@article{BenAmramYoffe2012,
  author  = {Amir M. Ben-Amram and Simon Yoffe},
  title   = {Corrigendum to ``A Simple and Efficient Union--Find--Delete Algorithm''},
  journal = {Theoretical Computer Science},
  volume  = {423},
  pages   = {75},
  year    = {2012},
  doi     = {10.1016/j.tcs.2011.12.074}
}

@inproceedings{AlstrupBenAmramRauhe1999,
  author    = {Stephen Alstrup and Amir M. Ben-Amram and Theis Rauhe},
  title     = {Worst-Case and Amortised Optimality in Union-Find},
  booktitle = {Proceedings of the Thirty-First Annual ACM Symposium on Theory of Computing},
  pages     = {499--506},
  publisher = {ACM},
  year      = {1999}
}

@article{BenAmramGalil2001,
  author  = {Amir M. Ben-Amram and Zvi Galil},
  title   = {A Generalization of a Lower Bound Technique Due to {Fredman} and {Saks}},
  journal = {Algorithmica},
  volume  = {30},
  number  = {1},
  pages   = {34--66},
  year    = {2001}
}

@inproceedings{FredmanSaks1989,
  author    = {Michael L. Fredman and Michael E. Saks},
  title     = {The Cell Probe Complexity of Dynamic Data Structures},
  booktitle = {Proceedings of the Twenty-First Annual ACM Symposium on Theory of Computing},
  pages     = {345--354},
  publisher = {ACM},
  year      = {1989}
}

@article{Blum1986,
  author  = {Norbert Blum},
  title   = {On the Single-Operation Worst-Case Time Complexity of the Disjoint Set Union Problem},
  journal = {SIAM Journal on Computing},
  volume  = {15},
  number  = {4},
  pages   = {1021--1024},
  year    = {1986}
}

@techreport{Smid1988,
  author      = {Michiel H. M. Smid},
  title       = {A Data Structure for the Union-Find Problem Having Good Single-Operation Complexity},
  institution = {Centre for Mathematics and Computer Science, Amsterdam},
  number      = {CT-1988-06},
  year        = {1988},
  note        = {Also appeared in ALCOM Algorithms Review 1 (1990)}
}

@article{Tarjan1975,
  author  = {Robert E. Tarjan},
  title   = {Efficiency of a Good but Not Linear Set Union Algorithm},
  journal = {Journal of the ACM},
  volume  = {22},
  number  = {2},
  pages   = {215--225},
  year    = {1975}
}

@article{TarjanVanLeeuwen1984,
  author  = {Robert E. Tarjan and Jan van Leeuwen},
  title   = {Worst-Case Analysis of Set Union Algorithms},
  journal = {Journal of the ACM},
  volume  = {31},
  number  = {2},
  pages   = {245--281},
  year    = {1984}
}

@phdthesis{Gortz2005,
  author = {Inge Li G{\o}rtz},
  title  = {Topics in Algorithms: Data Structures on Trees and Approximation Algorithms on Graphs},
  school = {IT University of Copenhagen},
  year   = {2005},
  month  = may
}

\end{document}